\documentclass[conference,letterpaper]{IEEEtran}

\usepackage[utf8]{inputenc} 
\usepackage[T1]{fontenc}
\usepackage{url}
\usepackage{ifthen}
\usepackage{cite,amssymb}
\usepackage[cmex10]{amsmath} % Use the [cmex10] option to ensure complicance
\usepackage[dvipsnames]{xcolor}
\usepackage{algorithmicx}
\usepackage{calc,pstricks, pgf, xcolor}
\usepackage{bbding}
\usepackage{bbm}
\usepackage{dsfont}
\usepackage{stfloats}
\usepackage{hyperref}
\usepackage{xparse}
\usepackage{enumerate}
\usepackage{subcaption}

\usepackage{algorithm}
\usepackage{algpseudocode}

\newtheorem{lemma}{Lemma}

\newenvironment{proof}[1][Proof]{\noindent\textbf{#1.} }{\ \rule{0.5em}{0.5em}}
\newcommand{\Ind}{\mathds{1}}

\newcommand{\ZZ}{\mathbb{Z}}
\newcommand{\RR}{\mathbb{R}}

\newcommand{\eps}{\varepsilon}

\newcommand{\Unif}{\mathrm{Uniform}}

\def\EE{\mathbb{E}}
\newcommand{\m}{\mathcal}

\NewDocumentCommand{\DIP}{e{^_}}{D^{\mathrm{IP}\IfValueT{#1}{,#1}}_{\IfValueT{#2}{#2}}}

\NewDocumentCommand{\DMM}{e{^_}}{D^{\mathrm{MM}\IfValueT{#1}{,#1}}_{\IfValueT{#2}{#2}}}

\DeclareMathOperator*{\argmin}{\arg\!\min}

\newcommand{\Or}[1]{\textcolor{cyan}{}}

\begin{document}
\title{High-Rate Nested-Lattice Quantized Matrix Multiplication with Small Lookup Tables} 
%\title{Linear Codes Transmitted Over a Binary Symmetric Channel: A Lower Bound on the Mutual Information via the Erasure Channel} 
% %%% Single author, or several authors with same affiliation:
% \author{%
%  \IEEEauthorblockN{Andrew R.~Barron}
%  \IEEEauthorblockA{Department of Statistics and Data Science\\
%                    Yale University\\
%                    New Haven, CT, USA\\
%                    Email: andrew.barron@yale.edu}
% }

%%% Several authors with up to three affiliations:
\author{%
  \IEEEauthorblockN{Iris~Kaplan}
 \IEEEauthorblockA{Computer Science and Engineering\\
                    Hebrew University of Jerusalem\\
                    Jerusalem, Israel\\
                    Email: iris.kaplan1@mail.huji.ac.il}  
    \and
      \IEEEauthorblockN{Or~Ordentlich}
  \IEEEauthorblockA{Computer Science and Engineering\\
                    Hebrew University of Jerusalem\\
                    Jerusalem, Israel\\
                    Email: or.ordentlich@mail.huji.ac.il}  
}

%%% Many authors with many affiliations:
% \author{%
%   \IEEEauthorblockN{Andrew R.~Barron\IEEEauthorrefmark{1},
%                     Claude E.~Shannon\IEEEauthorrefmark{2},
%                     David Slepian\IEEEauthorrefmark{2},
%                     and Jacob Ziv\IEEEauthorrefmark{2}\IEEEauthorrefmark{3}}
%   \IEEEauthorblockA{\IEEEauthorrefmark{1}%
%                    Department of Statistics and Data Science, Yale University, New Haven, CT, USA,
%                     andrew.barron@yale.edu}
%   \IEEEauthorblockA{\IEEEauthorrefmark{2}%
%                     Bell Telephone Laboratories, Inc.,
%                     Murray Hill, NJ, USA,
%                     \{csh,dsl,jz\}@bell-labs.com}
%   \IEEEauthorblockA{\IEEEauthorrefmark{3}%
%                     Department of Electrical Engineering, Technion---Institute of Technology, Haifa, Israel,
%                     jz@ee.technion.ac.il}
% }

\maketitle

%%%%%%
%% Abstract: 
%% If your paper is eligible for the student paper award, please add
%% the comment "THIS PAPER IS ELIGIBLE FOR THE STUDENT PAPER
%% AWARD." as a first line in the abstract. 
%% For the final version of the accepted paper, please do not forget
%% to remove this comment!
%%

\begin{abstract}
Recent work have shown that the quantization for matrix multiplication problem can be optimally solved by quantizing each column in each matrix using a nested lattice code, and then multiplying the de-quantized matrices. It was further demonstrated that when product codes of sub-dimension $d$ and rate $R$ are used, the de-quantization and inner product operations can be implemented with querying a lookup table (LUT) of size $2^{2dR}$, but this is only useful when $dR$ is sufficiently small. This in turn limits LUT-based inner product decoding to low-rate quantizers. In this work, we develop a rate $R$ hierarchical nested lattice quantization framework, which quantizes each vector to $M$ layers, and admits LUT-based inner product decoding using an LUT of size $2^{2d\frac{R}{M}}$, allowing for high-rate quantization. We provide analytic bounds on the loss of the developed scheme compared to standard nested lattice quantizers, and also numerically illustrate that this loss is negligible. Thus, our scheme enables to use small LUTs without compromising the overall distortion.
\end{abstract}

\section{Introduction}

Matrix multiplication constitutes the major part of the workload associated with inference in deep neural nets (DNNs) and large language models (LLMs). In order to compute $A^\top B$ for two matrices $A\in\RR^{n\times a}$ and $B\in \RR^{n\times b}$ and store the result, one needs to fetch/store $n(a+b)+ab$ entries from/to memory, whereas the computation requires $2nab$ operations. Modern hardware has become so efficient in performing multiplications and additions, that it is often the memory bandwidth that forms the bottleneck in matrix multiplication, especially when $a$ or $b$ are small. To remedy the limitations posed by the memory bandwidth, much research in the AI community over the last decade was dedicated to reducing the IO burden by compressing/quantizing the entries of one or both matrices $A,B$~\cite{gong2014compressing,jacob2018quantization,dettmers2022gpt3,yao2022zeroquant,xiao2023smoothquant,ma2024era,tseng2024quip,tseng2024qtip,quarot2024,guo2016dynamic,frantar2023optq}. See~\cite[Section I.A]{ordentlich2024optimal} for a discussion on what quantization rates are required for fully utilizing the compute cores. As it turns out, LLMs in the generation phase are typically memory-limited, and any reduction in the quantization rate speeds up the inference time proportionally. While on modern GPUs the memory limitation is so severe, that one can afford to spend a few cycles on de-quantization without affecting inference time, in modern CPUs the memory limitation is less acute, and often de-quantization must be highly efficient in order to result in shorter inference time.

Motivated by the above, this work develops a quantization scheme for matrix multiplication with fast decoding, which relies on using lookup tables (LUT) rather than first de-quantizing the elements of $A, B$ and computing their product. In particular, we rely on the nested-lattice based quantization for matrix multiplication scheme from~\cite{ordentlich2024optimal}, which was shown to be worst-case optimal when high-dimensional quantizers are used. A low-complexity variant of that scheme, based on the product of Voronoi codes~\cite{ConwaySloane83} in $\RR^d$, where $d$ is small, say $3\leq d\leq  8$, was further developed in~\cite{ordentlich2024optimal}, and was numerically shown to perform quite close to the fundamental limit, and attain state-of-the-art results for quantized LLMs~\cite{NestQuant}. While the encoding in this scheme requires to perform nearest neighbor decoding to the base lattice $L\subset\RR^d$, it was pointed out that decoding can be performed via repeated access to an LUT with $2^{2dR}$ entries, that stores all possible inner products between two vectors in the rate-$R$ Voronoi code. Decoding using LUTs is highly appealing (at least on a CPU), but only if the LUT is small enough to be stored in the fastest cache (L1 cache), which poses a significant constraint on the dimension-rate product $dR$ (say, $dR\leq 8-9$ for modern CPUs).

The goal of this work is to circumvent this limitation, and allow for nested-lattice quantization for matrix multiplication which admits efficient LUT-based decoding even for high quantization rate $R$. To this end, we develop an hierarchical nested lattice quantization scheme. Our scheme quantizes each vector in $\RR^d$ to $M\geq 1$ layers, and decoding of the inner product between two quantized vectors in $\RR^d$ only requires $M^2$ queries to a single LUT with $2^{2d\frac{R}{M}}$ entries. We show that despite the reduction of the required LUT size, our scheme's performance is extremely close to that of a Voronoi code with the same rate.

A Python implementation for quantized matrix multiplication using our scheme is available in~\cite{IrisPy}, whereas 
an efficient C implementation is available in~\cite{OriC}.

\textbf{Related Work:} Our construction falls within the framework of successive refinement~\cite{equitz1991successive}, which is also often referred to as \emph{embedded codes} or \emph{residual vector quantization} in the literature. The quadratic Gaussian rate-distortion problem is known to be successively refinable, that is, successive refinement codes can achieve the optimal rate-distortion tradeoff. Several image compression algorithms, such as Embedded Zerotree Wavelet  (EZW)~\cite{shapiro1993embedded} and Set Partitioning in Hierarchical Trees (SPIHT)~\cite{said1996new}, include a successive refinement component where first the MSBs of certain coefficients are sent, and the LSBs are sent afterwards. In an attempt to extend the MSB-to-LSB "scalar-quantization" refinement used by these algorithms to lattice-based vector quantization, ~\cite{mukherjee2002successive} and~\cite{fuchs2013embedded} constructed lattice-based successive refinement schemes in the same spirit as ours. However,~\cite{fuchs2013embedded} only considered layers with rate $R=1$ bit, whereas the encoding/decoding in~\cite{mukherjee2002successive} did not fully exploit the algebraic structure of nested lattice codes. Furthermore, the encoding in~\cite{mukherjee2002successive,fuchs2013embedded} is top-bottom, as is usually the case for successive refinement: one first quantizes the source to the coarsest lattice, and then uses a finer lattice to quantize the residual quantization error, and so on. Our scheme, on the other hand, quantizes the source to the finest lattice, and describes the obtained point as a coset of a coarse lattice whose index grows with the number of layers. As a consequence, we are able to obtain rigorous bounds on the performance of our hierarchical scheme compared to a single-layer Voronoi code of the same rate.

Using a Cartesian product of low-dimensional quantizers for quantizing a high-dimensional vector is a standard practice from the first days of digital communication~\cite{gray1998quantization},~\cite[Chapter 12.7]{GershoGrayBook}.
More recently, product codes for fast computation of approximate inner product or Euclidean distance were heavily studied within the context of approximate nearest neighbor (ANN) search and information retrieval. The idea of partitioning vectors to small chunks, building a quantizer for each chunk, and constructing corresponding LUTs for fast inner product computations was developed in~\cite{jegou2010product}. The work~\cite{jegou2010product} inspired a huge body of follow-up work. Among the many extensions studied, the combination of product codes with additive quantization seems to be closest to the approach we take in this paper. In additive quantization~\cite{babenko2014additive} one constructs the quantizer $Q:\RR^d\to [\prod_{m=1}^M K_m]$ as the Minkowski sum $\m{C}_1+\cdots+\m{C}_M$ of $M$ codebooks $\m{C}_1,\ldots,\m{C}_M$ of sizes $K_1,\ldots,K_M$, respectively. Typically, those codebooks are learned from the data using variants of Loyd's algorithm/K-means, and are therefore unstructured, which makes the encoding/decoding quite challenging. Moreover, in general product+additive quantization requires storing many different codebooks and LUTs. Our scheme uses the same lattice for all chucks and all sub-codebooks, and a \emph{single} lookup table for all operations involved in the inner product computation. For codebooks learned via $K$-means, the entries of the LUTs must be quantized as well, which may degrade performance~\cite{blalock2017bolt}.
In contrast, for standard choices of the base lattice, e.g., the $D_n$ or $E_n$ families, the possible inner products will further be integer-valued, so each entry of the LUT can be efficiently stored. Product codes with LUT based inner product decoding for machine learning applications was considered in~\cite{blalock2021multiplying}.

\section{Hierarchical Nested-Lattice Quantizers}

%\subsection{Lattice preliminaries}
We first review some basic lattice definitions. See~\cite{ramiBook} for a comprehensive treatment of lattices in information theory. For a lattice $L\subset\RR^d$ we define the nearest neighbor quantizer $Q_L:\RR^d\to L$ as
\begin{align}
Q_{L}(x)=\argmin_{\lambda\in L}\|x-\lambda\|, 
\label{eq:QLdef}
\end{align}
where ties are broken arbitrarily, but in a systematic manner. The Voronoi region $\m{V}_L$ is defined as the set of all points in $\RR^d$ that are closer to $0$ than to any other lattice point
\begin{align}
\m{V}=\m{V}_L=\left\{x\in\RR^d~:~Q_L(x)=0\right\}.    
\label{eq:Vorodef}
\end{align}
Any lattice $L\subset \RR^d$ has a (non-unique) generating matrix $G\in\RR^{d\times d}$ such that $L=G\ZZ^d$. %The covolume of the lattice $L$, denoted $\mathrm{covol}(L)$, is the volume of its Voronoi region (or any other fundamental cell of $L$), which is also equal to $|\det G|$. 
Let $Z\sim\Unif(\m{V}_L)$ be a random vector uniformly distributed over the Voronoi region of $L$. We define the second moment of the lattice $L$ as $\sigma^2(L)=\frac{1}{d}\EE\|Z\|^2$.
% \begin{align}
% \sigma^2(L)=\frac{1}{d}\EE\|Z\|^2.    
% \end{align}
% and the covariance matrix of $L$ as
% \begin{align}
% R(L)=\EE[Z Z^\top].    
% \end{align}
% The modulo operation with respect to the lattice $L$, is defined in this paper as
% \begin{align}
% [x]\bmod L=x-Q_{L}(x). 
% \end{align}
% Note that $[x]\bmod L\in\m{V}_L$. 
For any natural number $r$ we have that $r L\subset L$ and we construct the nested lattice quantizer/Voronoi code~\cite{ConwaySloane83}
\begin{align}
\m{A}_r=L\cap(r\m{V}),    
\end{align}
which is equivalent to the the quotient group $L/rL\cong (\ZZ/r\ZZ)^d$. This algebraic structure enables to use Voronoi codes as quantizers with fast encoding and decoding schemes. See~\cite{ConwaySloane83} and also~\cite[Algorithms 1 and 2] {ordentlich2024optimal}. In particular, these algorithms encode a vector $x\in\RR^d$ to $d\log_2 r$ bits, and based on those bits, the decoder outputs $\hat{x}\in\m{A}_r$, where $\hat{x}=Q_L(x)$ whenever $Q_L(x)\in r\m{V}$. The event $\hat{x}\neq Q_L(x)$ is called an \emph{overload} event.  The quantization rate of this scheme is $R=\log_2(r)$. 

Note that for Voronoi codes, the tie-breaking in~\eqref{eq:QLdef}, which affects the boundary of $\m{V}$ defined in~\eqref{eq:Vorodef}, is highly important, since for an even integer $r$ we will always have points of $L$ on the boundary of $r\m{V}$. To circumvent the numerical difficulties associated with treating the boundary, one can fix some very small and unstructured\footnote{In particular, we would like to avoid the situation where $\eps$ is contained in one of the hyperplanes defining the boundary of $\m{V}$.} perturbation vector $\eps\in\RR^d$ and replace the nearest neighbor quantizer $Q_L(x)$ with $Q_L(x+\eps)$ everywhere.

Our end goal in this paper is to compute inner products based on Voronoi codes and LUTs. To that end we can pre-compute all inner products of pairs of points $\lambda_i,\lambda_j\in\m{A}_r$ and store them in an LUT of size $r^{2d}B=2^{2dR}B$-bytes, assuming we use $B$ bytes to represent the value of each inner product. In order to allow fast access to the LUT, it must be small enough to fit in the L1 cache. This constrains the product $dR$, and typical numbers (depending on the processing unit that is used) are $dR\leq 8-9$. Our goal is to facilitate the use of Voronoi codes with LUT-based inner product decoding, while allowing arbitrarily large quantization rate $R$. This will be enabled via hierarchical nested-lattice quantizers.

\subsection{Proposed Hierarchical Nested-Lattice Quantizers}

\begin{algorithm}
\caption{Hierarchical Nested-Lattice Encoder}\label{alg:cap}
\begin{algorithmic}
\State \textbf{Inputs:} $x \in \mathbb{R}^d$, Lattice $L \subset \mathbb{R}^d$ with generating matrix $G \in \mathbb{R}^{d \times d}$, nesting ratio $q \in \mathbb{N}$, hierarchy depth $M \in \mathbb{N}$
\State \textbf{Outputs:} Encoding vectors $b_0, b_1, \dots, b_{M-1} \in [q]^d$
\State $\tilde{g} \gets x$ 
% \State $b_{\text{list}} \gets []$ 
\For{$m = 0$ to $M-1$}
    \State $\tilde{g} \gets Q_L(\tilde{g})$ 
    \State $b_m \gets [G^{-1} \cdot \tilde{g}] \mod q$ 
    % \State Append $b_i$ to $b_{\text{list}}$
    \State $\tilde{g} \gets \tilde{g} / q$
\EndFor
\State $\mathrm{OverloadError}=\Ind\{Q_L(\tilde{g})\neq 0\}$ 
\State \Return $b_0,b_1,\ldots,b_{M-1}$
\end{algorithmic}
\end{algorithm}

\begin{algorithm}
\caption{Hierarchical Nested Lattice Decoder}\label{alg:decode}
\begin{algorithmic}
\State \textbf{Inputs:} Encoding vectors $b_0, b_1, \dots, b_{M-1} \in [q]^d$, Lattice $L \subset \mathbb{R}^d$ with generating matrix $G \in \mathbb{R}^{d \times d}$, nesting ratio $q \in \mathbb{N}$, hierarchy depth $M \in \mathbb{N}$
\State \textbf{Output:} Reconstructed vector $\hat{x} \in L$
% \State $x_{\text{hat\_list}} \gets []$ 
\State $\hat{x}\gets 0$
\For{$m=0,\ldots,M-1$}
    \State $x_m \gets G \cdot b_m - q\cdot  Q_{L}((G \cdot b_m)/q)$
    \State $\hat{x} \gets \hat{x}+ q^{m} x_m$
\EndFor
\State \Return $\hat{x}$
\end{algorithmic}
\end{algorithm}

For a lattice $L\subset\RR^d$ and a natural number $q$, we denote by $Q_{qL}(\cdot)$ the nearest neighbor quantizer for the lattice $qL$. Note that $Q_{qL}(x)=q\cdot Q_L(x/q)$.
% \begin{align}
% Q_{qL}(x)=q\cdot Q_L(x/q).    
% \end{align}
For a non-negative integer $m$ and $x\in\RR^d$, define 
\begin{align}
Q^{\circ m}(x)&=Q_{L,q}^{\circ m}(x)=\left(Q_{q^m L}\circ Q_{q^{m-1} L}\cdots \circ Q_L\right)(x)\nonumber\\
&=Q_{q^m L}\left(Q_{q^{m-1}L}\left(\cdots\left(Q_L(x)\right)\right)\right).
\end{align}
Note that $Q^{\circ m}(x)$ is not equal to $Q_{q^m L}(x)$ in general, unless $L,q$ satisfy the perfect tiling condition: $(L\cap q\m{V})+\m{V}=q\m{V}$. The perfect tiling condition is met by $L=\ZZ$ with odd $q$, but is rarely met by any lattice in dimensions $d>1$.

Fix a lattice $L$, and two natural numbers $q,M$, and recall that $\m{A}_q=L\cap(q\m{V})$. Our scheme encodes each $x\in\RR^{d}$ to $d\cdot M\log_2(q)$ bits, and based on those bits the decoder outputs a point $\hat{x}$ in the constellation $\m{C}_{L,q,M}=\sum_{m=0}^{M-1} q^{m}\m{A}_q    $
% \begin{align}
% \m{C}_{L,q,M}=\sum_{m=0}^{M-1} q^{m}\m{A}_q    
% \end{align}
where the sum above is a Minkowski sum. Whenever $Q_L(x)\in \m{C}_{L,q,M}$ we have that the reconstruction produced by our scheme satisfies $\hat{x}=Q_L(x)$. Figure~\ref{fig:codebooks} depicts the constellation $\m{C}_{L,q,M}$ for the hexagonal lattice $L=A_2$, $q=6$ and $M=3$.

% \begin{figure*}[!tbp]
%   \centering
%   \begin{subfigure}[b]{0.45\textwidth}
%     \includegraphics[width=\textwidth]{A2q6M2.png}
%     \caption{$M=2$}
%     \label{fig:codebook_M2}
%   \end{subfigure}
%   \hfill
%   \begin{subfigure}[b]{0.45\textwidth}
%     \includegraphics[width=\textwidth]{A2q6M3.png}
%     \caption{$M=3$}
%     \label{fig:codebook_M3}
%   \end{subfigure}
%   \caption{Codebooks of the hierarchical nested lattice $A_2$ quantizer with $q=4$ and $r_{q,M}$ as defined below. To avoid points on the boundaries of the Voronoi cells, each input vector $\mathbf{x}$ is perturbed with a random small dither drawn from a Gaussian distribution. The figures show the results for different values of $M$.}
%   \label{fig:codebooks}
% \end{figure*}

To describe our scheme, define $\tilde{g}_m(x)=\frac{Q^{\circ m}(x)}{q^m}$,
% \begin{align}
% \tilde{g}_m(x)=\frac{Q^{\circ m}(x)}{q^m},    
% \end{align}
and note that $\tilde{g}_m\in L$ by definition, and can be computed recursively as
\begin{align}
\tilde{g}_0(x)=Q_L(x);~~~\tilde{g}_{m}(x)=Q_L\left(\frac{\tilde{g}_{m-1}(x)}{q}\right),~~m=1,\ldots.\nonumber
\end{align}
For $m=0,\ldots,M-1$, let 
\begin{align}
g_m(x)&=Q^{\circ m}(x)-Q^{\circ (m+1)}(x)\nonumber\\
&=Q^{\circ m}(x)-Q_{q^{m+1}L}\left(Q^{\circ m}(x) \right)\nonumber\\
&=q^m\left(\frac{Q^{\circ m}(x)}{q^m}-Q_{qL}\left(\frac{Q^{\circ m}(x)}{q^m}\right) \right)\\
&=q^m\left(\tilde{g}_m(x)-Q_{qL}\left(\tilde{g}_m(x) \right)\right).\label{eq:mthlayer}
\end{align}
Since $\tilde{g}_m(x)\in L$, we clearly have that $g_m(x)\in q^m(L\cap q\m{V})=q^m \m{A}_q$. Furthermore, we can use the standard Voronoi encoder/decoder~\cite{ConwaySloane83},~\cite[Algorithms 1 and 2]{ordentlich2024optimal} for representing $\tilde{g}_m(x)-Q_{qL}\left(\tilde{g}_m(x) \right)\in\m{A}_q$ using $d\log_2(q)$ bits, and mapping those bits back to $\m{A}_q$. Algorithm~\ref{alg:cap} implements the encoder in our scheme. It recursively computes $\{\tilde{g}_m\}_{m=0}^{M-1}$ and represents each of them using the sequences $\{b_m\}_{m=0}^{M-1}$, each of length $d\log_2(q)$ bits. The corresponding decoder is given in Algorithm~\ref{alg:decode}. It recovers from each vector $b_m$ the corresponding point $\tilde{g}_m(x)-Q_{qL}\left(\tilde{g}_m(x)\right)$ in $\m{A}_q$, and outputs the reconstruction
\begin{align}
\hat{x}=\sum_{m=0}^{M-1} q^m \left[\tilde{g}_m(x)-Q_{qL}\left(\tilde{g}_m(x)\right) \right]=\sum_{m=0}^{M-1} g_m(x). 
\label{eq:telescope}
\end{align}

\begin{lemma}
Let $x\in\RR^d$ and let $\hat{x}\in L$ be its reconstruction using the hierarchical nested-lattice quantizer, that is 
the output of Algorithm~\ref{alg:decode} applied on the output of Algorithm~\ref{alg:cap}.  Then $\hat{x}=Q_L(x)$ iff $Q^{\circ M}(x)=0$.
\label{lem:overload}
\end{lemma}

\begin{proof}
Telescoping $\sum_{m=0}^{M-1} g_m(x)$ in~\eqref{eq:telescope}, we observe that
\begin{align}
\hat{x}=Q^{\circ 0}(x)-Q^{\circ M}(x)=Q_L(x)-Q^{\circ M}(x).    
\end{align}
Thus, $\hat{x}=Q_L(x)$ iff $Q^{\circ M}(x)=0$.
\end{proof}

As a consequence of Lemma~\ref{lem:overload}, we see that the binary variable $\mathrm{OverloadError}$ computed in Algorithm~\ref{alg:cap} indeed satisfies $\mathrm{OverloadError}=\Ind\{\hat{x}\neq Q_L(x)\}$.

\textbf{Scaling and dithering:} In order to get the smallest distortion using the proposed scheme, as well as when using standard Voronoi codes, one scales the constellation by a factor $\beta>0$. The granular error is $\sigma^2(\beta L)=\beta^2\sigma^2 (L)$ and is increasing in $\beta$. On the other hand, the overload probability is decreasing in $\beta$. Thus, one typically looks for the smallest $\beta$ for which the overload probability is ``small enough''. Often, one also uses a \emph{dither} vector $z\in\m{V}$ to shift the constellation. Specifically, to implement scaling by $\beta>0$ and dithering by $z\in \m{V}$, we set the input to our encoder as $\frac{x}{\beta}-z$ instead of $x$, and the output of our decoder as $\beta (\hat{x}+ z)$ instead of $\hat{x}$.

\textbf{Overload avoidance mechanism:} We employ an overload avoidance mechanism similar to the one introduced in~\cite{ordentlich2024optimal}. Specifically, we first set $\beta=\beta_0$ and $T=0$. We input $\frac{x}{\beta}-z$ to our encoder, and check whether $\mathrm{OverloadError}=0$. If so, we send the encoded bits $b_0,\ldots,b_{M-1}$. Otherwise, we set $\beta\gets 2^\alpha\beta$, $T\gets T+1$, where $\alpha>0$ is a parameter of the algorithm, and try again, and so on until $\mathrm{OverloadError}=0$ and we send the encoded bits $b_0,\ldots,b_{M-1}$. We also send an entropy-coded description of $T$ to the decoder, using $\approx H(T)$ bits. In total, the expected rate of this scheme is $M\log_2(q)+\frac{H(T)}{d}$. The decoder, in turn, reconstructs $T$ from the entropy coded bits and $\hat{x}$ from $b_0,\ldots,b_{M-1}$, and outputs $2^{\alpha T} \beta_0 (\hat{x}+z)$.

\textbf{Successive Refinement:} In some situations one may wish to start by recovering the source $x$ with low resolution and gradually improve the resolution as more bits describing it become available. If we run the for-loop in Algorithm~\ref{alg:decode} from $m=M-1$ down to $m=0$, we will get a gradually improving description. Moreover, if we only recover $t<M$ layers, by running that for-loop from $i=M-1$ down to $i=M-t$, we will recover $Q^{\circ (M-t)}(x)-Q^{\circ M}(x)$, which equals $Q^{\circ (M-t)}(x)$ if overload did not occur.

\subsection{Bounds and numerical results}

\begin{figure}[!tbp]
  \centering    \includegraphics[width=0.45\textwidth]{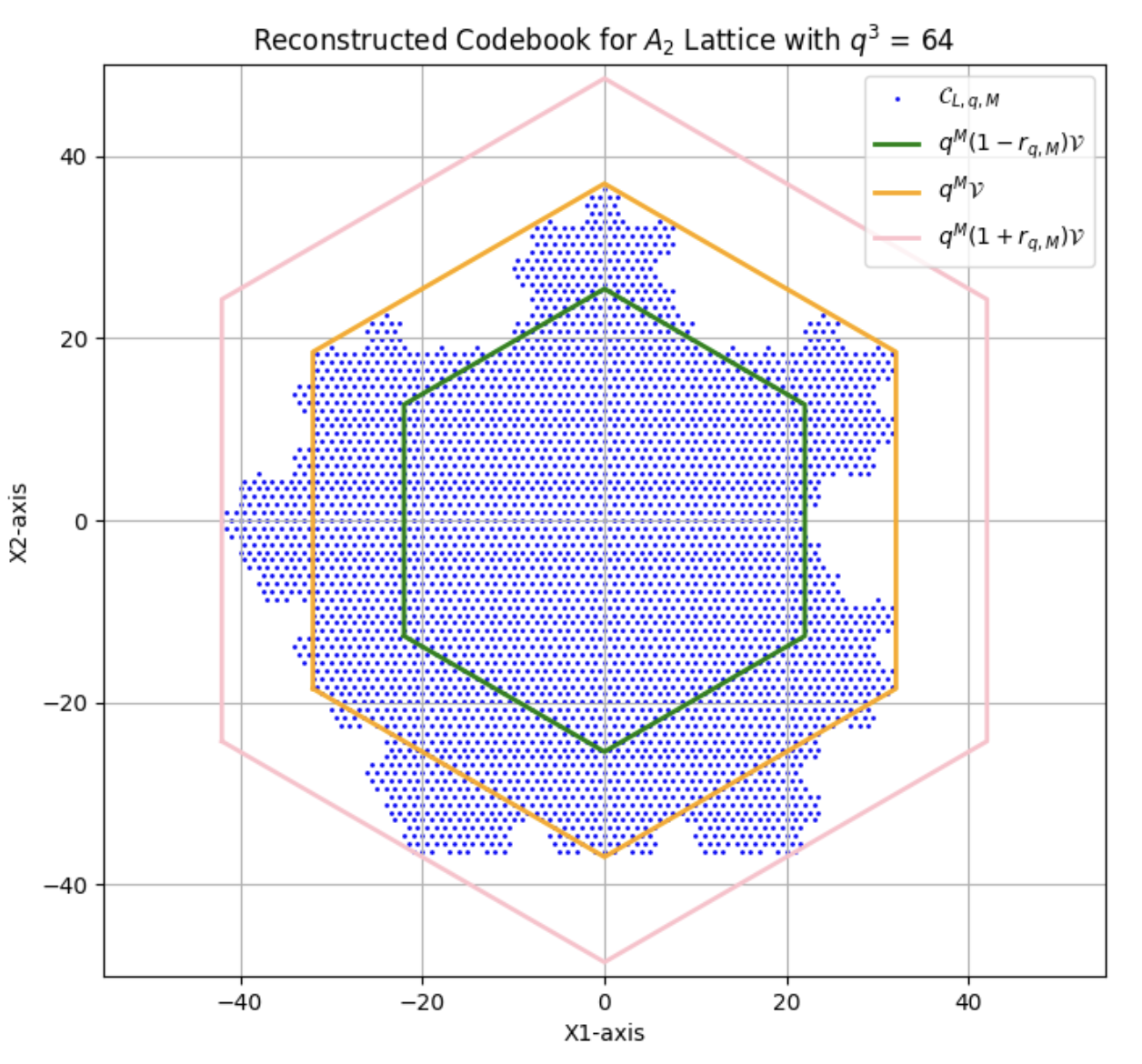}
    %\label{fig:codebook_M3}
  %\end{subfigure}
  \caption{Codebook of the hierarchical nested lattice quantizer with $L=A_2$, $q=4$ and $M=3$.}
  \label{fig:codebooks}
\end{figure}

Let $\m{P}_{q,M}=\{x\in\RR^d~:~Q^{\circ M}(x)=0\}$,
% \begin{align}
% \m{P}_{q,M}=\{x\in\RR^d~:~Q^{\circ M}(x)=0\},    
% \end{align}
and note that $\m{P}_{q,M}$ is a fundamental cell of the lattice $q^M L$. Note that furthermore $\m{C}_{L,q,M}=L\cap \m{P}_{q,M}$, and it therefore follows that $\m{C}_{L,q,M}\cong L/q^M L$. A Voronoi code $\m{A}_{q^M}$ selects the coset representatives with the lowest energy, and consequently approximately minimizes the overload probability among all choices of coset representatives of $L/q^M L$. While the hierarchical scheme described in the previous subsection has several advantages over a Voronoi code of the same rate and same lattice $L$, it selects the coset representatives of $L/q^M L$ as $L\cap \m{P}_{q,M}$, and is therefore inferior to the corresponding Voronoi code in terms of overload probability. The following result shows that the overload probability of $\m{C}_{L,q,M}$ is nevertheless upper (resp. lower) bounded by the overload probability of a Voronoi code whose rate is $\log\left(1\mp \frac{1-q^{1-M}}{q-1}\right)$ bits smaller (resp. greater).

\begin{lemma}
\begin{align}
\m{A}_{q^M(1-r_{q,M})} \subset \m{C}_{L,q,M}\subset \m{A}_{q^M(1+r_{q,M})},
\label{eq:regioninclusion}
\end{align}    
where $r_{q,M}=q^{-M}\sum_{m=1}^{M-1} q^m=\frac{1-q^{1-M}}{q-1}$.
\end{lemma}

The green and pink scaled Voronoi regions in Figure~\ref{fig:codebooks} illustrate the inclusion in~\eqref{eq:regioninclusion}.

\begin{proof}
Recall that since $\m{V}$ is a convex set in $\RR^d$, we have that $\alpha \m{V}+\beta \m{V}=(\alpha+\beta)\m{V}$ for all $\alpha,\beta>0$.
To prove that $\m{C}_{L,q,M}\subset \m{A}_{q^M(1+r_{q,M})}$, note that    
\begin{align}
\m{C}_{L,q,M}&=\sum_{m=0}^{M-1} q^m (L\cap q\m{V})\subset L\cap \left(q\m{V}\cdot\sum_{m=0}^{M-1} q^m \right) \nonumber\\
&=L\cap ((q^M+\sum_{m=1}^{M-1} q^m)\m{V}).   
\end{align}
To prove $\m{A}_{q^M(1-r_{q,M})} \subset \m{C}_{L,q,M}$, it suffices to show that for any $y\in \m{A}_{q^M(1-r_{q,M})}$ it holds that $Q^{\circ M}(y)=0$.  We have,
\begin{align}
&Q^{\circ {(M-1)}}(y)=Q_{q^{M-1} L}\left(Q^{\circ {(M-2)}}(y) \right)\nonumber\\
&\in  Q^{\circ {(M-2)}}(y)+q^{M-1} \m{V}
\in \cdots\in  Q_L(y)+\sum_{m=1}^{M-1} q^m \m{V}\label{eq:Qcircinclusion}
\end{align}
and since $y\in \m{A}_{q^M(1-r_{q,M})}$, we also have $y=Q_L(y)\subset q^M(1-r_{q,M})\m{V}$. Consequently,
\begin{align}
Q^{\circ {(M-1)}}(y)\in q^M(1-r_{q,M})\m{V}+\sum_{m=1}^{M-1} q^m \m{V}=q^M \m{V},    
\end{align}
which implies that $Q^{\circ M}(y)=Q_{q^M L}\left(Q^{\circ {(M-1)}}(y)\right)=0$
\end{proof}

\textbf{Simulation results:} We plot the distortion-rate tradeoff attained by three different nested lattice quantization schemes: (a) Voronoi code with $r=q^M$ (b)Voronoi code with $r=q^M(1-r_{q,M})$ (c)The developed hierarchical scheme with $M$ layers and nesting ratio $q$. For all schemes we use the same base lattice $L$, and use the overload avoidance mechanism described above, with $\alpha=1/3$ and $\beta_0$ optimized separately for each of the three schemes. In the simulations we use $L=D_4$, $M=2$, and $q\in\{3,4,\ldots,9\}$. We quantize $N=5000$ iid realizations of a $\m{N}(0,I_4)$ source, and plot the obtained distortion-rate tradeoff for each scheme in Figure~\ref{fig:DRnoProd}. We also plot the Shannon limit $D(R)=2^{-2R}$ for reference. It can be seen that the hierarchical nested-lattice quantizer is indeed strictly better than the (lower-rate) Voronoi code with $r=q^M(1-r_{q,M})=q(q-1)$, and is almost as good as the reference Voronoi code with $r=q^M=q^2$. Remarkably, using the overload avoidance mechanism, both schemes are less than $1/2$ bit away from the fundamental limit, using a simple code of dimension $d=4$.

% \begin{figure}
%     \centering
%     \includegraphics[width=1\linewidth]{Distortion-rate with overload and f_beta for D4 lattice.png}
%     \caption{Distortion-Rate curves for nested lattice quantizers.}
%     \label{fig:DRnoProd}
% \end{figure}

\section{Fast Inner-Product Computation}

\begin{figure*}
  \centering
  \begin{subfigure}[b]{0.45\textwidth}
    \includegraphics[width=\textwidth]{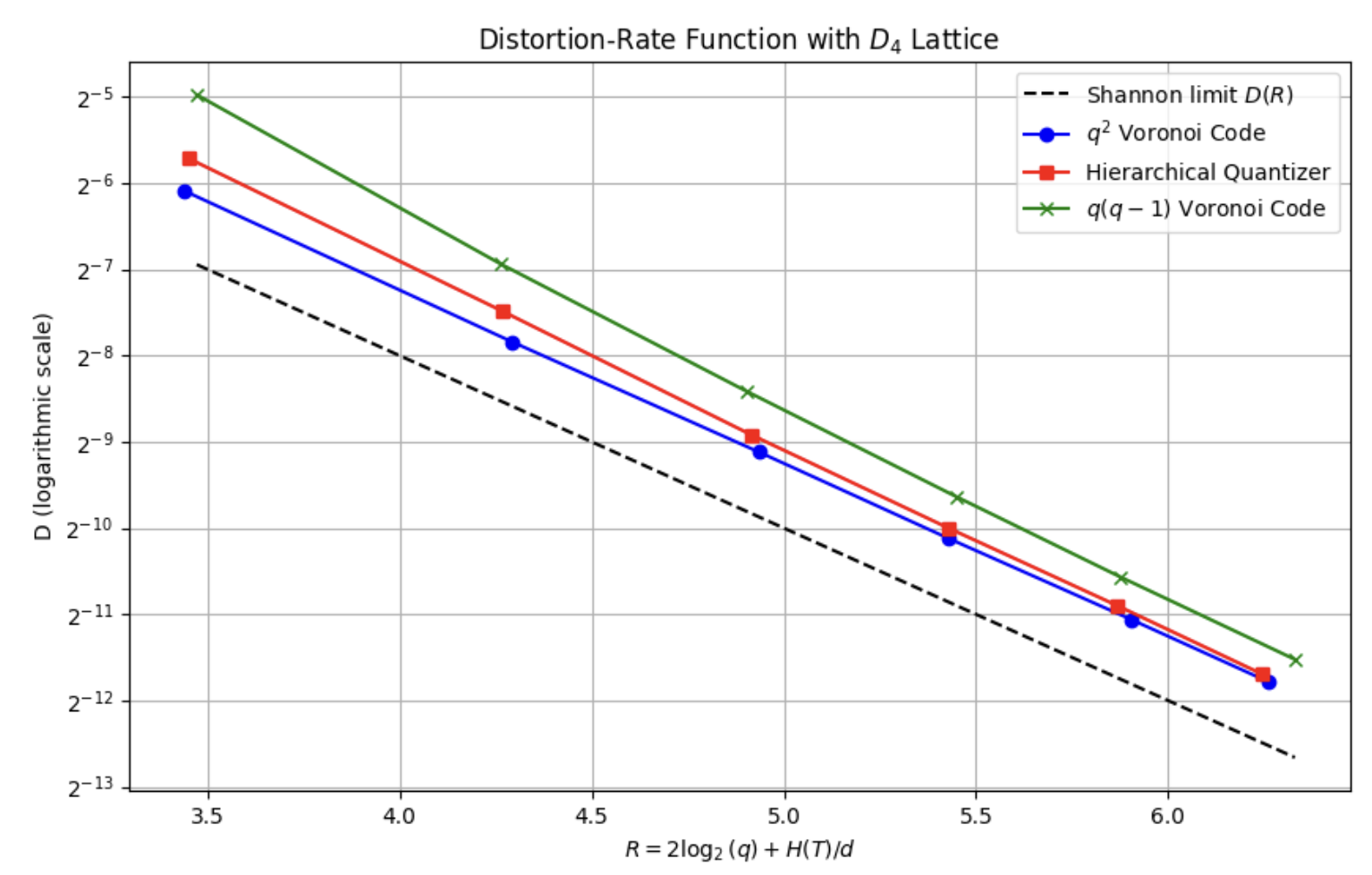}
    \caption{Vector quantization, $M=2$, $q$ varies}
    \label{fig:DRnoProd}
  \end{subfigure}
  \hfill
  \begin{subfigure}[b]{0.45\textwidth}
    \includegraphics[width=\textwidth]{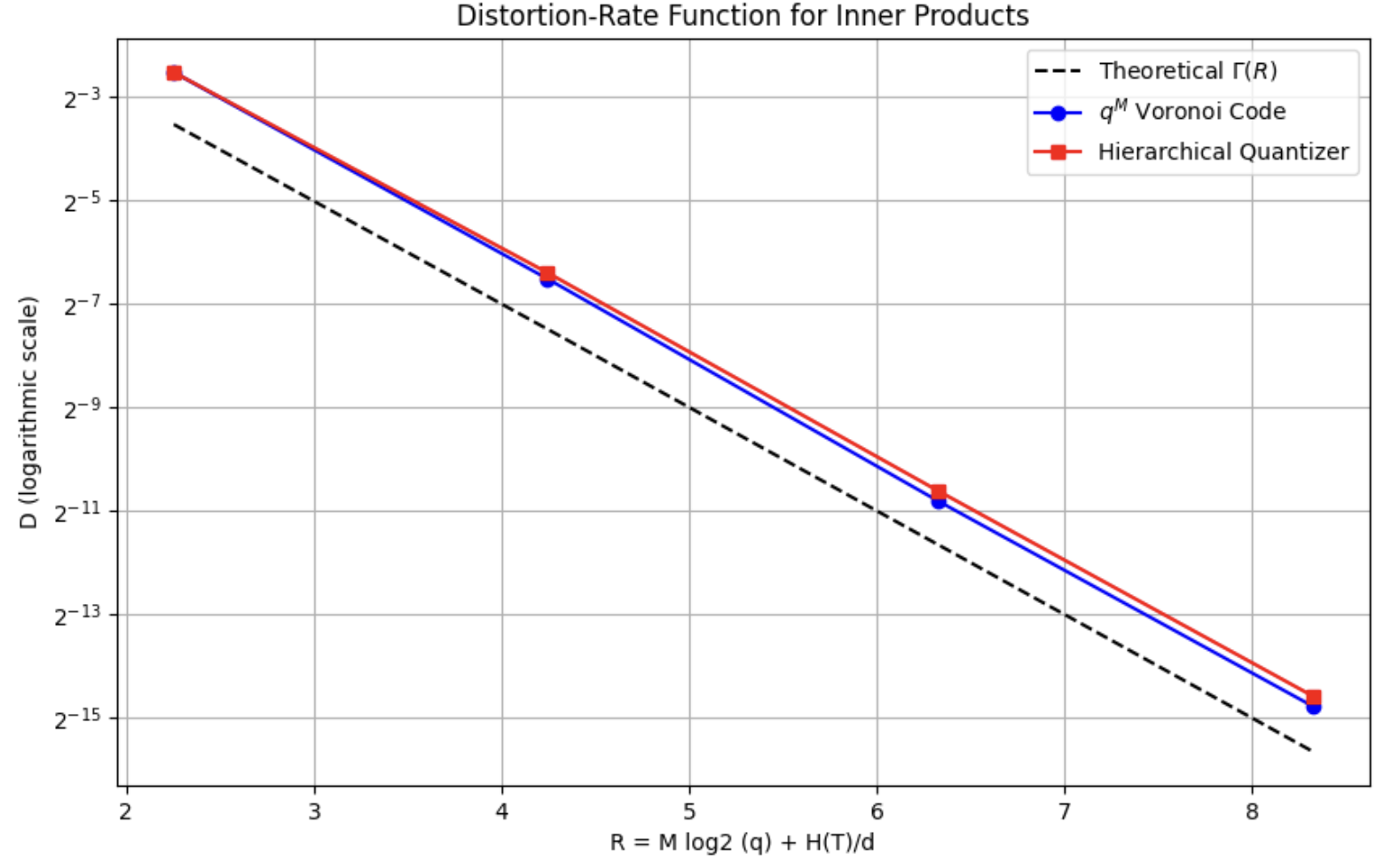}
    \caption{Inner product quantization, $q=4$, $M$ varies}
    \label{fig:DRIP}
  \end{subfigure}
  \caption{Distortion-Rate curves for nested lattice quantizers, $L=D_4$.}
  \label{fig:codebooks}
\end{figure*}

 Our main motivation for introducing the hierarchical scheme was to develop a fast decoder for the quantization for inner product computation problem. Let $x$ and $y$ be two vectors in $\RR^d$ that were quantized by the hierarchical scheme above. We have that
\begin{align}
\hat{x}=\sum_{i=0}^{M-1}q^i \hat{x}_i;~~\hat{y}=\sum_{j=0}^{M-1}q^j \hat{y}_j,   
\end{align}
where $\hat{x}_i,
\hat{y}_j\in\m{A}_q$ for all $i,j\in\{0,\ldots,M-1\}$. Consequently,
\begin{align}
\hat{x}^\top \hat{y}=\sum_{i=0}^{M-1}\sum_{j=0}^{M-1}q^{i+j}\cdot \hat{x}_i^\top \hat{y}_j. 
\end{align}
Note further that all inner products $\hat{x}_i^\top \hat{y}_j$ participating in the sum above, involve two vectors in $\m{A}_q$. Furthermore, the vector $\hat{x}_i$ is represented by the encoding vector $b_i(x)\in[q]^d$ in Algorithm~\ref{alg:cap} (in particular, we have $\hat{x}_i=G\cdot b_i(x)-q\cdot Q_L(G\cdot b_i(x))/q)$), and similarly $\hat{y}_j$ is represented by the encoding vector $b_j(y)\in[q]^d$. We can therefore pre-compute all $q^{2d}$ inner products in $\m{A}_q\times\m{A}_q$ and store the results in an LUT $\{\m{L}(b_i,b_j)\}_{b_i,b_j\in[q]^d\times [q]^d}$ indexed by two encoding vectors $b_i$, $b_j$. It therefore follows that for if our goal is to quantize $x,y\in\RR^d$ in order to compute an approximation for $x^\top y$, we can quantize each of them to $b_0(x),\ldots,b_{M-1}(x)$ and $b_0(y),\ldots,b_{M-1}(y)$ using Algorithm~\ref{alg:cap} with rate $R=M\log_2(q)$ bits per entry, and then decode the inner product via
\begin{align}
\hat{x}^\top \hat{y}= \sum_{i=0}^{M-1}\sum_{j=0}^{M-1}q^{i+j}\cdot \m{L}(b_i(x),b_j(y)).
\label{eq:LUTip}
\end{align}
The complexity of decoding the inner product $\hat{x}^\top \hat{y}$ therefore consists of querying the LUT $M^2$ times, $M^2$ scalar multiplications of the fetched LUT values by $q^{i+j}$ (if $q$ is a power of $2$ this can be implemented by simple bit-shifts) and $M^2$ additions. As mentioned above, the main gain of the hierarchical scheme is that we only need to store a single LUT of size $2^{d\log_2(q)}=2^{2d\frac{R}{M}}$, instead of an LUT of size $2^{2dR}$ that is needed if standard Voronoi codes are used. The ``price'' for the reduction of the LUT size is that we need to access it $M^2$ times.

\textbf{Scaling and dithering:} As in the previous section, one can scale and shift the constellation by encoding $\frac{x}{\beta_1}-z_1$ rather than $x$ and $\frac{y}{\beta_2}-z_2$ rather than $y$, where $\beta_1,\beta_2>0$ and $z_1,z_2\in\m{V}$. The decoder in turn, should output $\beta_1\beta_2(\hat{x}+z_1)^\top(\hat{y}+z_2)$ rather than $\hat{x}^\top \hat{y}$. Taking arbitrary dither vectors $z_1,z_2\in\m{V}$ does not permit to decode $(\hat{x}+z_1)^\top(\hat{y}+z_2)$ only by accessing the LUT $\m{L}$, as in~\eqref{eq:LUTip}. To circumvent this issue, we restrict the dither vectors to the constellation $q^{-1}\m{A}_q\subset\m{V}$. In particular, we choose two vectors $b_{z_k}\in[q]^d$, $k=1,2$ and set $z_k=q^{-1}\left[G\cdot b_{z_k}-q\cdot Q_L((G\cdot b_{z_k})/q)\right]$ as our dither vectors. We can then define $b_{-1}(x)=b_{z_1}$ and $b_{-1}(y)=b_{z_2}$ and compute 
\begin{align}
(\hat{x}+z_1)^\top (\hat{y}+z_2)= \sum_{i=-1}^{M-1}\sum_{j=-1}^{M-1}q^{i+j}\cdot \m{L}(b_i(x),b_j(y)).    
\end{align}

\textbf{One-sided quantization:} There are many practical scenarios where one needs to compute many inner products $y^\top x_k$, $k=1,\ldots,K$ between a fixed vector $y\in\RR^d$ and many other vectors $x_1,\ldots,x_K\in\RR^d$, $k\gg 1$. For instance, $x_1,\ldots,x_K$ can be vectors in a database,  and $y$ is a query for which one needs to find the approximate nearest neighbor in the database. In such situations, it is often the case that $y$ can be stored essentially in full resolution, but the vectors in the database must be quantized. Assuming we quantize each $x_k$ (assuming no dithering for simplicity) using our hierarchical scheme, we can compute each inner product as
\begin{align}
 y^\top \hat{x}_k=\sum_{i=0}^{M-1}  q^i \m{L}_y(b_i(x_k)),~~k=1,\ldots,K,
\end{align}
where $\m{L}_y$ is an LUT of size $q^d=2^{d\frac{R}{M}}$ consisting of the values of all inner products between $y$ and a vector in $\m{A}_q$.

\textbf{Arbitrary dimension via product codes:} In order to solve the quantization for inner product computation problem for vectors $x,y\in\RR^n$, $n\gg 1$, we use a product of quantizers for $\RR^d$. Let $K=n/d$, and assume for simplicity that $K$ is an integer. We may split $x$ and $y$ to $K$ chunks, each of size $d$, such that $x=[x^\top_1|\cdots|x^\top_K]^\top$, $y=[y^\top_1|\cdots|y^\top_K]^\top$. We can then quantize each chunk $x_k$ (resp. $y_k$) to $\hat{x}_k$ (resp. $\hat{y}_k$) using our hierarchical $d$-dimensional scheme, and decode the inner product as
\begin{align}
\hat{x}^\top \hat{y}=\sum_{k=1}^K \hat{x}_k^\top\hat{y}_k. 
\label{eq:IPviaPQ}
\end{align}

\textbf{Universality via random rotation:} Our hierarchical quantizer is tailored to a Gaussian iid source, and may not perform well when the source to be quantized does not resemble a Gaussian vector. In order to obtain a universal quantization for inner product scheme, whose performance depends only on $\|x\|_2$, $\|y\|_2$ and $|x^\top y|$, one should draw a random rotation matrix $S\in\RR^{n\times n}$, and quantize $S x/\|x\|_2$ (resp. $S y/\|y\|_2$) instead of $x$ (reps. $y$). The inner product between the quantized vectors should then be scaled back by $\|x\|_2 \|y\|_2$. See~\cite{ordentlich2024optimal} for details and and analysis.

\textbf{Simulation results:} We draw $N=5000$ pairs of iid vectors $X,Y\sim\m{N}(0,I_n)$ where $n=512$. We use $d=4$. Hence, for each pair, we quantize each vector by chunking it to $K=512/4=128$ pieces, quantize each piece using a lattice quantizer based on $L=D_4\subset\RR^4$ with  (a)Voronoi codes with $r=q^M$ where $q=4$ and $M$ varies (b)Hierarchical nested-lattice quantizer with $q=4$ and varying $M$, and computing the inner product $\hat{X}^\top\hat{Y}$ as in~\eqref{eq:IPviaPQ}. For the hierarchical scheme, we use the LUT approach~\eqref{eq:LUTip}. For both schemes we use the overlaod avoidance mechanism described above, with $\alpha=1/3$.
We define the distortion as $D=\frac{1}{n}\EE(X^\top Y-\hat{X}^\top \hat{Y})^2$, and also plot the fundamental limit $D\geq \Gamma(R)$ from~\cite{ordentlich2024optimal} (for $R>0.906$ we have $\Gamma(R)=2\cdot 2^{-2R}-2^{-4R})$. The results are plotted in Figure~\ref{fig:DRIP} for $M=1,2,3,4$. It is evident that the hierarchical scheme has performance very close to that of Voronoi codes with the same rate, and that both schemes are about half a bit away of the fundamental limit.

% \begin{figure}
%     \centering
%     \includegraphics[width=1\linewidth]{D(R)IP.png}
%     \caption{Distortion-Rate curves for nested lattice quantizers.}
%     \label{fig:DRIP}
% \end{figure}

%%%%%%
%% Appendix:
%% If needed a single appendix is created by
%%
%\appendix
%%
%% If several appendices are needed, then the command
%%
% \appendices
%%
%% in combination with further \section commands can be used.
%%%%%%

\section*{Acknowledgment}

This work was supported by the ISRAEL SCIENCE 
FOUNDATION (ISF),
grant No.1641/21. 
%%%%%%
%% To balance the columns at the last page of the paper use this
%% command:
%%
%\enlargethispage{-1.2cm} 
%%
%% If the balancing should occur in the middle of the references, use
%% the following trigger:
%%

%%
%% which triggers a \newpage (i.e., new column) just before the given
%% reference number. Note that you need to adapt this if you modify
%% the paper.  The "triggered" command can be changed if desired:
%%
%\IEEEtriggercmd{\enlargethispage{-20cm}}
%%
%%%%%%

%%%%%%
%% Appendix:
%% If needed a single appendix is created by
%%
%\appendix
%%
%% If several appendices are needed, then the command
%%
% \appendices
%%
%% in combination with further \section commands can be used.
%%%%%%

%%%%%%
%% References:
%% We recommend the usage of BibTeX:
%%

\newpage
 %\IEEEtriggeratref{11}
\bibliographystyle{IEEEtran}
\bibliography{NestedLUTbib}

\end{document}